\documentclass{IEEEtran}

\IEEEoverridecommandlockouts

\usepackage{cite}
\usepackage{amsmath,amssymb,amsfonts,amsthm,newtxmath,mleftright}
\usepackage{graphicx,xcolor}
\usepackage[caption=false,font=footnotesize]{subfig}
\usepackage{fancyhdr}

\renewcommand{\vec}[1]{\ensuremath{\mathbf{#1}}} 
\newcommand{\mat}[1]{\ensuremath{\mathbf{#1}}} 
\newcommand{\T}{\ensuremath{\mathsf{T}}} 

\newcommand{\expval}[1]{\operatorname{E}[#1]}

\newcommand{\erfc}{\operatorname{erfc}}
\newcommand{\var}[1]{\operatorname{var}[#1]}

\newcommand{\pd}{\ensuremath{p_\mathit{d}}}
\newcommand{\pfa}{\ensuremath{p_\mathit{fa}}}

\newtheorem{proposition}{Proposition}

\theoremstyle{definition}
\newtheorem{definition}[proposition]{Definition}

\theoremstyle{remark}
\newtheorem*{remark}{Remark}

\begin{document}
	
\title{A Family of Neyman-Pearson-Based Detectors for Noise-Type Radars}

\author{
	\IEEEauthorblockN{David Luong,~\IEEEmembership{Graduate Student Member, IEEE}, Bhashyam Balaji,~\IEEEmembership{Senior Member, IEEE}, and \\ Sreeraman Rajan,~\IEEEmembership{Senior Member, IEEE}}
	
	\thanks{D.\ Luong is with Carleton University, Ottawa, ON, Canada K1S 5B6. Email: david.luong3@carleton.ca.}
	\thanks{B.\ Balaji is with Defence Research and Development Canada, Ottawa, ON, Canada K2K 2Y7. Email: bhashyam.balaji@drdc-rddc.gc.ca.}
	\thanks{S.\ Rajan is with Carleton University, Ottawa, ON, Canada K1S 5B6. Email: sreeraman.rajan@carleton.ca.}
}

\maketitle

\begin{abstract}
	We derive a detector that optimizes the target detection performance of any single-input single-output noise radar satisfying the following properties: it transmits Gaussian noise, it retains an internal reference signal for matched filtering, all external noise is additive white Gaussian noise, and all signals are measured using heterodyne receivers. This class of radars, which we call noise-type radars, includes not only many types of standard noise radars, but also a type of quantum radar known as quantum two-mode squeezing radar. The detector, which we derive using the Neyman-Pearson lemma, is not practical because it requires foreknowledge of a target-dependent correlation coefficient that cannot be known beforehand. (It is, however, a natural standard of comparison for other detectors.) This motivates us to study the family of Neyman-Pearson-based detectors that result when the correlation coefficient is treated as a parameter. We derive the probability distribution of the Neyman-Pearson-based detectors when there is a mismatch between the pre-chosen parameter value and the true correlation coefficient. We then use this result to generate receiver operating characteristic curves. Finally, we apply our results to the case where the correlation coefficient is small. It turns out that the resulting detector is not only a good one, but that it has appeared previously in the quantum radar literature.
\end{abstract}


\section{Introduction}

Quantum radar, particularly the subclass of quantum radars known as \emph{quantum illumination} (QI) radars, has been studied from a theoretical point of view for over a decade now \cite{lloyd2008qi,tan2008quantum,zhuang2017neyman,wilde2017gaussian,zhuang2017optimum,lopaeva2013qi,zhang2014qi,balaji2018qi,torrome2021quantum}. More recently, a renewed wave of interest in quantum radar arose when Wilson and his team announced the results of the world's first experiment to validate the concept of a microwave quantum radar \cite{chang2018quantum,luong2019roc}. (Note that this was not itself a quantum radar because it contained \emph{extra} components: amplifiers that broke the entanglement in the microwave signals.) What was more, the main conclusions of this experiment by Wilson et al.\ were shown to be repeatable by another experimental group \cite{barzanjeh2019experimental}. In particular, \cite{luong2019roc} showed that a quantum radar could potentially reduce the integration time required to detect a target by at least a factor of 8.

The radar in the Wilson experiment was not a standard QI radar because it did not implement the quantum measurements envisioned in \cite{tan2008quantum}. Instead, it used standard heterodyne detectors. For this reason, the radar architecture that was validated by the Wilson experiment was described as a ``quantum-enhanced noise radar'' in \cite{chang2018quantum} and named \emph{quantum two-mode squeezing} (QTMS) radar in \cite{luong2019roc}. These characterizations differentiate QTMS radar from standard QI, but at the same time draw a connection between quantum radars and \emph{noise radars}, which---while somewhat niche---have been studied in the radar literature for a long time \cite{cooper1967random,lukin1998millimeter,narayanan2004design,thayaparan2006noise,lukin2008Ka,tarchi2010SAR,kulpa2013signal,narayanan2016noise,savci2019trials,wasserzier2019noise,savci2020noise}. This connection between noise radar and QTMS radar was fleshed out in later papers from our group \cite{luong2019cov,luong2020int}. In particular, we showed that noise radars and QTMS radars lie on a continuum parameterized by a correlation coefficient $\rho$ that we will describe in the following section.

In this paper, we derive an upper bound for the detection performance of any radar satisfying the following properties: its transmit signal consists of Gaussian noise, it retains a reference signal for matched filtering (assumed to be jointly Gaussian with the transmit signal), all external noise is additive white Gaussian noise, and all signals are measured using heterodyne detectors. We will call such radars \emph{noise-type radars}. QTMS radars fall under this description, as do some types of noise radar \cite{dawood2001roc}. This upper bound takes the form of a detector based on the Neyman-Pearson lemma, which guarantees that no other detector can outperform this one (as quantified by the receiver operating characteristic curve). This detector is only a theoretical construct because it requires \emph{a priori} knowledge of the correlation coefficient $\rho$, which in most cases cannot be known beforehand. It forms, however, a natural basis for comparison with more practical detectors. 

More generally, we consider the family of Neyman-Pearson detectors obtained by varying $\rho$, and explore their statistical properties when there is a mismatch between the $\rho$ chosen for the detector and the true $\rho$. It turns out that the mismatch does not lead to drastically different detection performance. Finally, we consider the case $\rho \to 0$. In this case, the resulting Neyman-Pearson-based detector is not only a viable detector, but it has already been studied previously. This connects the results of this paper to previous work done in \cite{luong2019roc,luong2020simdet,barzanjeh2019experimental}.

\section{Noise Radar Target Detection}
\label{sec:nr_target_det}

QTMS radars and (simple) noise radars generate two electromagnetic signals consisting of Gaussian noise. One signal is transmitted through free space, while the other signal is retained as a reference for matched filtering. In this paper, we will derive an upper bound on the detection performance of such radars. To that end, we now formally state the conditions under which our results are valid:
\begin{itemize}
	\item The radar is a single-input single-output radar.
	\item The transmit and reference signals are jointly Gaussian white noise processes.
	\item The received and reference signals are measured using heterodyne receivers.
	\item All system and external noise, including thermal noise and atmospheric noise, can be modeled as additive white Gaussian noise (AWGN).
\end{itemize}
It is important that we specify the use of heterodyne detection, because standard QI does not use heterodyne receivers \cite{tan2008quantum,guha2009gaussian} and therefore do not fall under the analysis in this paper.

Because each electromagnetic signal can be decomposed into in-phase and quadrature components, the two signals must be described by \emph{four} real-valued time series. We will denote the in-phase and quadrature voltages of the received signal by $I_1[n]$ and $Q_1[n]$, respectively, and the corresponding voltages for the reference signal by $I_2[n]$ and $Q_2[n]$. We will assume the use of digitizers, so the four time series are discrete and indexed by $n$. 

For our purposes, we do not need to introduce any notation for the transmit signal. The reference signal contains all the information about the transmit signal that is available to the radar, and can be thought of as a copy of the transmit signal. Due to quantum noise, this copy is necessarily imperfect, and this imperfection becomes prominent at low signal powers. However, the use of quantum entanglement can help mitigate the effect of quantum noise. This is one motivation for quantum radar \cite{luong2020magazine}.

Under the above conditions, $I_1[n]$, $Q_1[n]$, $I_2[n]$, and $Q_2[n]$ are jointly Gaussian with zero mean and are pairwise independent unless the time lag is zero. We will therefore assume the time lag to be zero and drop the index $n$. Because the multivariate Gaussian distribution is fully specified by its mean and covariance matrix, the signals are completely described once we give the $4 \times 4$ covariance matrix $\expval{\vec{x}\vec{x}^\T}$, where $\vec{x} = [I_1, Q_1, I_2, Q_2]^\T$. According to \cite{luong2019cov}, the covariance matrix for QTMS radars is
\begin{equation} \label{eq:QTMS_cov}
	\expval{\vec{x}\vec{x}^\T} =
	\begin{bmatrix}
		\sigma_1^2 \mat{1}_2 & \rho \sigma_1 \sigma_2 \mat{R}'(\phi) \\
		\rho \sigma_1 \sigma_2 \mat{R}'(\phi)^\mathsf{T} & \sigma_2^2 \mat{1}_2
	\end{bmatrix}
\end{equation}
where $\sigma_1^2$ and $\sigma_2^2$ are the received and reference signal powers, respectively, $\rho$ is a coefficient satisfying $0 \leq \rho \leq 1$, $\phi$ is the phase shift between the signals, $\mat{1}_2$ is the $2 \times 2$ identity matrix, and $\mat{R}'(\phi)$ is the reflection matrix 
\begin{equation}
	\mat{R}'(\phi) = 
	\begin{bmatrix}
		\cos \phi & \sin \phi \\
		\sin \phi & -\cos \phi
	\end{bmatrix} \! .
\end{equation}
Standard noise radars are described by the same covariance matrix except that the rotation matrix
\begin{equation}
	\mat{R}(\phi) = 
	\begin{bmatrix}
		\cos \phi & \sin \phi \\
		-\sin \phi & \cos \phi
	\end{bmatrix}
\end{equation}
replaces the reflection matrix.

Because we wish to derive an upper bound on QTMS/noise radar detection performance, we make the following simplifying assumptions:
\begin{itemize}
	\item The phase shift between the received and reference signals is known, and can therefore be set to zero.
	\item The signal powers are assumed known, so we may normalize the signals to have unit power.
\end{itemize}
It should be obvious that the two assumptions made here cannot decrease the detection performance of the radar. Under these assumptions, the covariance matrix takes on the simplified form
\begin{equation} \label{eq:cov_simplified}
	\expval{\vec{x}\vec{x}^\T} = \mat{\Sigma}(\rho) =
	\begin{bmatrix}
		1 & 0 & \rho & 0 \\
		0 & 1 & 0 & \pm\rho \\
		\rho & 0 & 1 & 0 \\
		0 & \pm\rho & 0 & 1 
	\end{bmatrix} \! .
\end{equation}
Note that $\mat{\Sigma}(\rho)$ is a function of the parameter $\rho$ alone. This parameter, which is of great importance for target detection, quantifies the correlation between the transmit and reference signals and will therefore be called simply the \emph{correlation coefficient}. The positive sign in \eqref{eq:cov_simplified} applies to standard noise radars, while the negative sign refers to QTMS radars. The choice of sign does not affect our analysis.

We note here that the true distinction between QTMS radars and standard noise radars lies not in the sign chosen in \eqref{eq:cov_simplified}, but in the magnitude of $\rho$. Simply put, QTMS radars are able to achieve higher values of $\rho$ than is possible using standard radar engineering techniques \cite{luong2020int}. This makes it easier to distinguish between the presence and absence of targets, a problem which we will now consider.

\subsection{Hypothesis Testing}

In the absence of clutter, the received and reference signals will only be correlated if a target is present. Therefore, the noise radar target detection problem requires us to test the following hypotheses:
\begin{equation} \label{eq:hypotheses}
	\begin{alignedat}{3}
		H_0&: \rho = 0 &&\quad\text{Target absent} \\
		H_1&: \rho > 0 &&\quad\text{Target present.}
	\end{alignedat}
\end{equation}
To the best of our knowledge, there exists no optimal detector function for this problem. In particular, the Neyman-Pearson lemma does not apply, nor does the Karlin-Rubin theorem on uniformly most powerful tests \cite{luong2022likelihood}. However, we can arbitrarily choose some fixed value of $\rho$, which we will denote by $\kappa$, and consider the problem of distinguishing the following \emph{simple} hypotheses:
\begin{equation} \label{eq:hypotheses_simple}
	\begin{aligned}
		H_0&: \rho = 0 \\
		H_1&: \rho = \kappa
	\end{aligned}
\end{equation}
In this case, the Neyman-Pearson lemma does apply. And if it so happens that the correlation coefficient for a given target is indeed $\kappa$, the Neyman-Pearson detector will be optimal; this is guaranteed by the Neyman-Pearson lemma. However, the same detector could conceivably be used for targets with different values of $\rho$. The detector will not be optimal for such targets, but it may still be a viable detector. This motivates us to consider the family of Neyman-Pearson-based detectors obtained by treating $\kappa$ as a parameter.

It is important to understand that the hypotheses \eqref{eq:hypotheses_simple} constitute a mathematical abstraction which allows us to derive the family of Neyman-Pearson-based detectors. The hypotheses which must be distinguished by a radar are \eqref{eq:hypotheses}, because $\rho$ cannot be known beforehand. This is because prior knowledge of $\rho$ requires prior knowledge of every factor in the radar range equation \cite{luong2022performance}. Needless to say, it is unreasonable to require radar operators to know the range and the radar cross section of a target before detecting it.

That said, one motivation for analyzing the Neyman-Pearson family of detectors comes from the early QI literature. In order to exploit known results from the theory of quantum information, some of the pioneers in QI implicitly used a hypothesis test of the form \eqref{eq:hypotheses_simple}. In \cite{lloyd2008qi,tan2008quantum,zhuang2017neyman,wilde2017gaussian,zhuang2017optimum}, for example, one of the implicit assumptions is that the path between the radar's transmitter and receiver has a known \emph{transmittance}. In radar engineering terms, the authors assumed that the ratio of receive power to transmit power, $P_r/P_t$, is known beforehand. This is equivalent to \eqref{eq:hypotheses_simple} because the correlation coefficient is a function of the various parameters in the radar range equation, and the relationship between $P_r$ and $P_t$ is, by definition, the radar radar equation. However, some of the foundational results of QI are based on \eqref{eq:hypotheses_simple}, and it would be interesting to see how well these results would hold up when applied to the correct hypotheses \eqref{eq:hypotheses}.

We emphasize that the results in this paper do not apply directly to standard QI radars as described in \cite{lloyd2008qi,tan2008quantum}, because they do not use heterodyne detection and their detector outputs are not in-phase and quadrature voltages. This violates the assumptions listed earlier in this section. However, our results---which do apply to QTMS radar, a variant of QI radar---are indirect evidence that the results of the early QI literature may be partially applicable under the real-life hypotheses \eqref{eq:hypotheses}.

\section{The Neyman-Pearson Detector}

In this section, we will define the Neyman-Pearson detector and establish its basic properties, including its probability density function.

In the following, we use an overline for the sample mean over $N$ samples. For any random variable $Z$, we write
\begin{equation}
	\bar{Z} = \frac{1}{N} \sum_{n=1}^N z[n]
\end{equation}
where $z[1], \dots, z[N]$ are $N$ realizations of $Z$. We will sometimes refer to $N$ as the number of samples integrated by the radar.

\begin{definition}
	The \emph{Neyman-Pearson (NP) detector} is defined as
	\begin{equation} \label{eq:det_NP}
		\bar{D}_\kappa = \bar{D}_0 - \kappa \frac{\bar{P}_\text{tot}}{2}
	\end{equation}
	where
	\begin{subequations}
		\begin{gather}
			\label{eq:P_tot}
			P_\text{tot} \equiv I_1^2 + Q_1^2 + I_2^2 + Q_2^2, \\
			\label{eq:det0}
			D_0 \equiv I_1 I_2 \pm Q_1 Q_2,
		\end{gather}
	\end{subequations}
	and $0 \leq \kappa < 1$. The positive sign is to be used for standard noise radars, and the negative sign for QTMS radars.
\end{definition}

The very name ``Neyman-Pearson detector'' suggests that this detector was constructed to be optimal in the sense of the Neyman-Pearson lemma. We now prove that this is indeed so.

\begin{proposition} \label{prop:NP_optimal}
	The NP detector is optimal for the hypotheses \eqref{eq:hypotheses_simple} in the sense that it maximizes the probability of detection $\pd$ for a given $\pfa$.
\end{proposition}

\begin{proof}
	The log-likelihood ratio for the multivariate normal distribution is well known. For $N$ independently drawn samples, the log-likelihood for our particular case is
	\begin{equation} \label{eq:log_like}
		\begin{split}
			\ell(\rho) = -\frac{N}{2} \mleft[ \frac{\bar{P}_\text{tot} - \bar{D}_0 \rho}{1 - \rho^2} + 2 \ln (1 - \rho ^2) + 2 \ln(2 \pi) \mright]
		\end{split}
	\end{equation}
	The log-likelihood ratio is
	\begin{equation} \label{eq:log_like_ratio}
		-2[\ell(0) - \ell(\kappa)] = N \mleft[ \frac{2\bar{D}_0\kappa - \bar{P}_\text{tot}\kappa^2}{1 - \kappa^2} - 2\ln(1 - \kappa^2) \mright]
	\end{equation}
	The Neyman-Pearson lemma \cite[Theorem 8.3.12]{casella2002stat} tells us that, for the simple hypotheses \eqref{eq:hypotheses_simple}, the likelihood ratio test is the most powerful test for a given significance level. This is equivalent to the above statement about $\pd$ and $\pfa$, since $\pd$ is the power of the test and $\pfa$ is its significance level. Hence \eqref{eq:log_like_ratio} is optimal: it maximizes $\pd$ for a given $\pfa$.
	
	Finally, recall that the performance of a test statistic is invariant under strictly increasing transformations, as this amounts to nothing more than a reparameterization of the decision threshold. This allows us to obtain an equivalent test statistic by dropping the $2\ln(1 - \kappa^2)$ term and multiplying by an appropriate constant to obtain \eqref{eq:det_NP}. Therefore, the performance of the NP detector \eqref{eq:det_NP} is the same as that of \eqref{eq:log_like_ratio}, which we have already shown to be optimal.
\end{proof}

\begin{remark}
	This proposition implies that the detection performance of $\bar{D}_\rho$, the NP detector with $\kappa$ set to the correct value of $\rho$, is an upper bound for the detection performance of any radar satisfying the following conditions: the radar transmits a single electromagnetic signal consisting of Gaussian white noise while retaining a copy for matched filtering, all external noise is AWGN, and the radar employs heterodyne detection.
\end{remark}

Though the NP detector is of theoretical interest as a standard with which to compare other detectors, its optimality is of little practical utility. This is because, as discussed above, we cannot assume that we know the value of $\rho$ in advance, so we cannot know which $\kappa$ to choose. This motivates us to study is the performance of the NP detector when $\rho \neq \kappa$. From this viewpoint, $\bar{D}_\kappa$ can be considered as a \emph{family} of Neyman-Pearson-based detectors for \eqref{eq:hypotheses}, where we can choose $\kappa$ to be any desired value. However, there is no guarantee that $\bar{D}_\kappa$ is a good detector for any $\kappa$ when $\rho \neq \kappa$. One goal of this paper is to establish the properties of $\bar{D}_\kappa$ when the Neyman-Pearson lemma does not hold.

\subsection{Distribution of the NP Detector}

We will shortly prove that $\bar{D}_\kappa$ follows a probability distribution known as the \emph{variance-gamma distribution}. Therefore, we now define this distribution and quote some of its most important properties.

\begin{definition} \label{def:VG}
	A random variable $X$ follows a \emph{variance-gamma distribution} if its probability density function has the form
	\begin{align}
		f_\mathit{VG}(x | c, \theta, \sigma, \nu) &= \frac{ 2 \exp[ \theta(x - c)/\sigma^2 ] }{ \sigma \sqrt{2\pi} \nu^{1/\nu} \Gamma(1/\nu) } \mleft( \frac{ |x - c| }{ \sqrt{2\sigma^2/\nu + \theta^2} } \mright)^{ \frac{1}{\nu} - \frac{1}{2} } \nonumber \\
			&\qquad\phantom{=} \times K_{ \frac{1}{\nu} - \frac{1}{2} } \mleft( \frac{ \sqrt{2\sigma^2/\nu + \theta^2} }{\sigma^2} |x - c| \mright)
			\label{eq:VG_pdf}
	\end{align}
	where $c$, $\sigma$, $\theta$, and $\nu$ are real-valued parameters with $\sigma > 0$ and $\nu > 0$. We denote by $K_\alpha(\cdot)$ the modified Bessel function of the second kind of order $\alpha$. In this case, we write $X \sim \mathit{VG}(c, \sigma, \theta, \nu)$.
\end{definition}

Note that the variance-gamma distribution is parameterized in a variety of ways. We have followed the parameterization in \cite[Eq.\ (15)]{seneta2004fitting}. 

We will also need the characteristic function for the variance-gamma distribution. It is given in \cite[Eq.\ (16)]{seneta2004fitting}:
\begin{equation} \label{eq:VG_char_func}
	\varphi_\mathit{VG}(t) = e^{jct} \mleft( 1 - j \theta \nu t + \frac{\sigma^2 \nu}{2} t^2 \mright)^{-\frac{1}{\nu}}.
\end{equation}
It follows immediately that, if $X \sim \mathit{VG}(c, \sigma, \theta, \nu)$,
\begin{subequations}
	\begin{align}
		\label{eq:VG_mean}
		\expval{X} &= c + \theta \\
		\label{eq:VG_var}
		\var{X} &= \sigma^2 + \theta^2 \nu.
	\end{align}
\end{subequations}

With these results in hand, we may now prove the following proposition.

\begin{proposition}
	The NP detector follows a variance-gamma distribution:
	\begin{equation}
		\bar{D}_\kappa \sim \mathit{VG} \mleft( 0, \sqrt{ \frac{2(1-\rho^2)(1-\kappa^2)}{N} }, 2(\rho - \kappa), \frac{1}{N} \mright).
	\end{equation}
\end{proposition}

\begin{proof}
	First, note that the NP detector \eqref{eq:det_NP} is a quadratic form in $I_1$, $Q_1$, $I_2$, and $Q_2$. More explicitly,
	\begin{equation} \label{eq:diag_form}
		\bar{D}_\kappa = \overline{\vec{x}^\T \mat{A} \vec{x}}
	\end{equation}
	where
	\begin{equation}
		\mat{A} =
		\frac{1}{2}
		\begin{bmatrix}
			-\kappa & 0 & 1 & 0 \\
			0 & -\kappa & 0 & -1 \\
			1 & 0 & -\kappa & 0 \\
			0 & -1 & 0 & -\kappa
		\end{bmatrix} \! .
	\end{equation}
	With an appropriate change of basis $\vec{x} \to \mat{B}\vec{x}$, any quadratic form can be diagonalized via the transformation $\mat{A} \to (\mat{B}^{-1})^\T\mat{A}\mat{B}^{-1}$. Second, by the properties of the multivariate normal distribution, any change of basis such that $\mat{B} \mat{\Sigma}(\rho) \mat{B}^\T = \mat{1}_4$ (a whitening transformation) will result in a random vector whose components are independent and normally distributed. 
	
	Our derivation hinges on the following insight: there exists a whitening matrix $\mat{B}$ which also diagonalizes the quadratic form \eqref{eq:diag_form}. One such matrix is 
	\begin{equation}
		\mat{B} = 
		\frac{1}{\sqrt{2}} \!
		\begin{bmatrix}
			\frac{1}{\sqrt{1 + \rho}} & 0 & \frac{1}{\sqrt{1 + \rho}} & 0 \\
			0 & \frac{1}{\sqrt{1 + \rho}} & 0 & \frac{-1}{\sqrt{1 + \rho}} \\
			\frac{1}{\sqrt{1 - \rho}} & 0 & \frac{-1}{\sqrt{1 - \rho}} & 0 \\
			0 & \frac{1}{\sqrt{1 - \rho}} & 0 & \frac{1}{\sqrt{1 - \rho}}
		\end{bmatrix} \! .
	\end{equation}
	It can be verified that this choice of $\mat{B}$ indeed transforms the covariance matrix $\mat{\Sigma}(\rho)$ into the identity matrix. At the same time, $\mat{B}$ diagonalizes $\mat{A}$ and yields the following expression for $\bar{D}_\kappa$:
	\begin{equation}
		\bar{D}_\kappa = c_+(\overline{Z_1^2} + \overline{Z_2^2}) - c_-(\overline{Z_3^2} + \overline{Z_4^2}).
	\end{equation}
	Here, $\mat{B}\vec{x} = [Z_1, Z_2, Z_3, Z_4]^\T$ are independent standard normal variates and
	\begin{equation}
		c_\pm = \frac{(1 \pm \rho)(1 \mp \kappa)}{2}.
	\end{equation}
	We can simplify this further by noting that the sum of squares of standard normal random variables is chi-squared distributed:
	\begin{equation}
		\bar{D}_\kappa = c_+ U_1 - c_- U_2
	\end{equation}
	where $U_1, U_2 \sim \chi^2_{2N}$. 
	
	The remainder of the proof is straightforward. The characteristic function for a chi-squared random variable with $k$ degrees of freedom is $\varphi_{\chi}(t) = (1 - 2jt)^{-k/2}$. Using the fact that the characteristic function of a linear combination of random variables is the product of rescaled versions of their individual characteristic functions, we find that
	\begin{equation} \label{eq:det_NP_char_func}
		\varphi_{\mathit{NP}}(t) = \mleft( 1 - \frac{2 j (\rho - \kappa)}{N} t + \frac{4 c_+c_-}{N^2} t^2 \mright)^{-N} .
	\end{equation}
	Because characteristic functions uniquely specify their probability distributions, the proposition follows upon comparing \eqref{eq:det_NP_char_func} with \eqref{eq:VG_char_func}.
\end{proof}

\begin{figure}[t]
	\centerline{\includegraphics[width=\columnwidth]{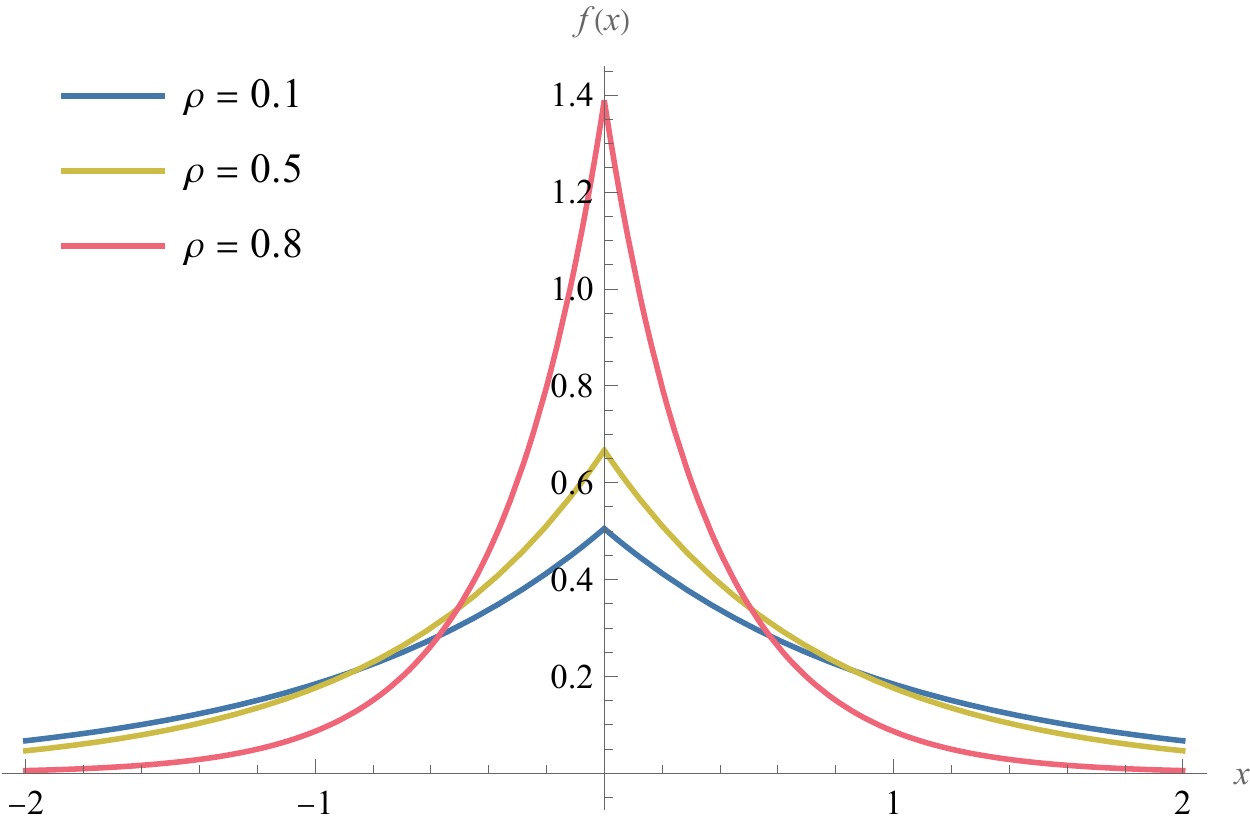}}
	\caption{PDFs of the NP detector when $N = 1$, $\rho \in \{0.1, 0.5, 0.8\}$, and $\kappa = \rho$.}
	\label{fig:PDF_optimal_rho}
\end{figure}

\begin{figure}[t]
	\centerline{\includegraphics[width=\columnwidth]{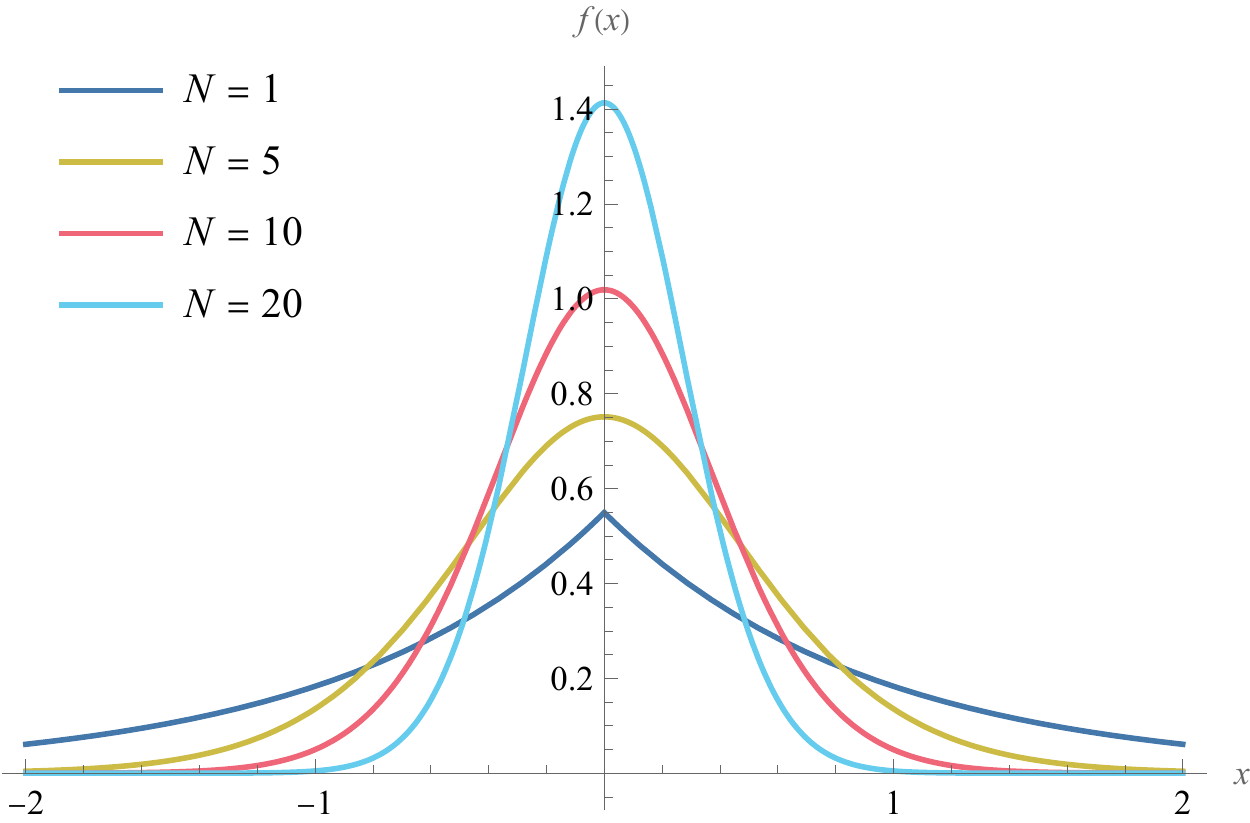}}
	\caption{PDFs of the NP detector when $\kappa = \rho = 0.3$ and $N \in \{1, 5, 10, 20\}$.}
	\label{fig:PDF_optimal_N}
\end{figure}

\begin{figure}[t]
	\centerline{\includegraphics[width=\columnwidth]{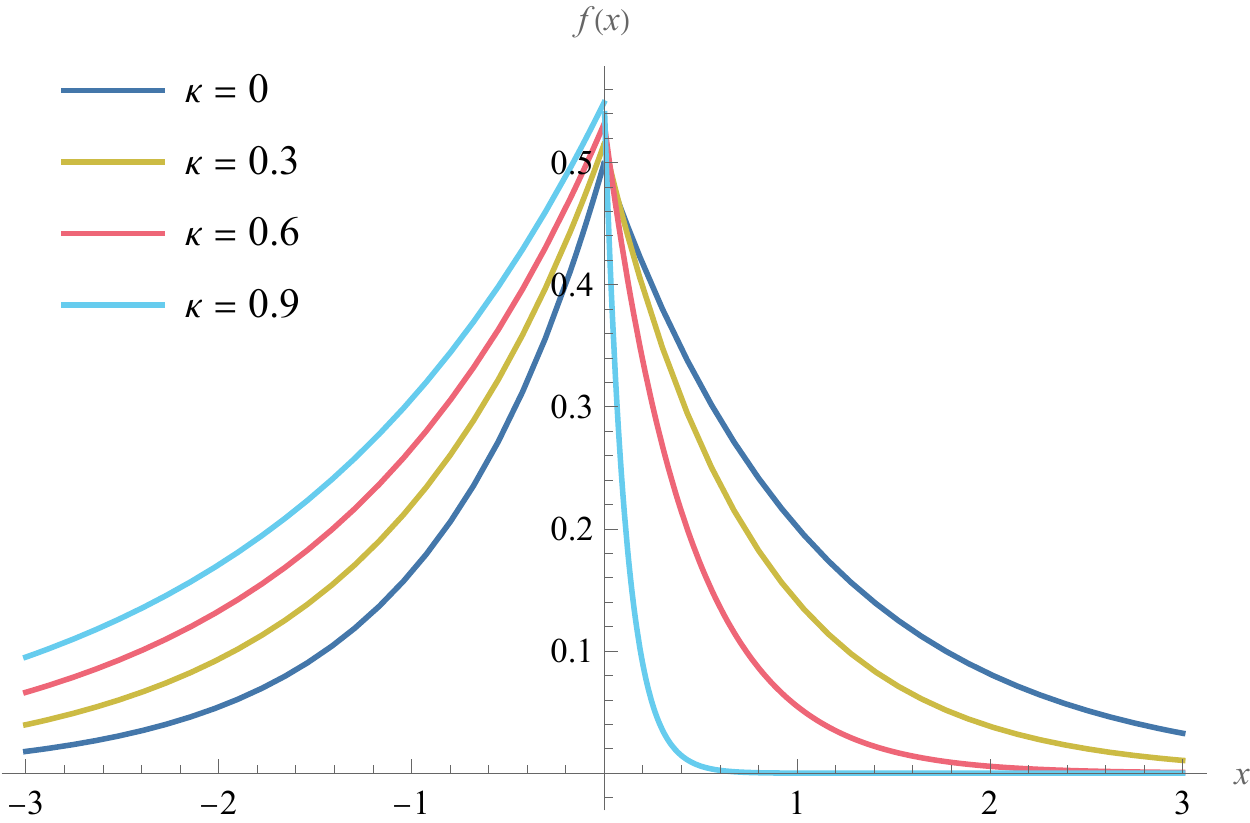}}
	\caption{PDFs of the NP detector when $N = 1$, $\rho = 0.1$, and $\kappa \in \{0, 0.3,\linebreak[1] 0.6, 0.9\}$.}
	\label{fig:PDF_kappa_1}
\end{figure}

\begin{figure}[t]
	\centerline{\includegraphics[width=\columnwidth]{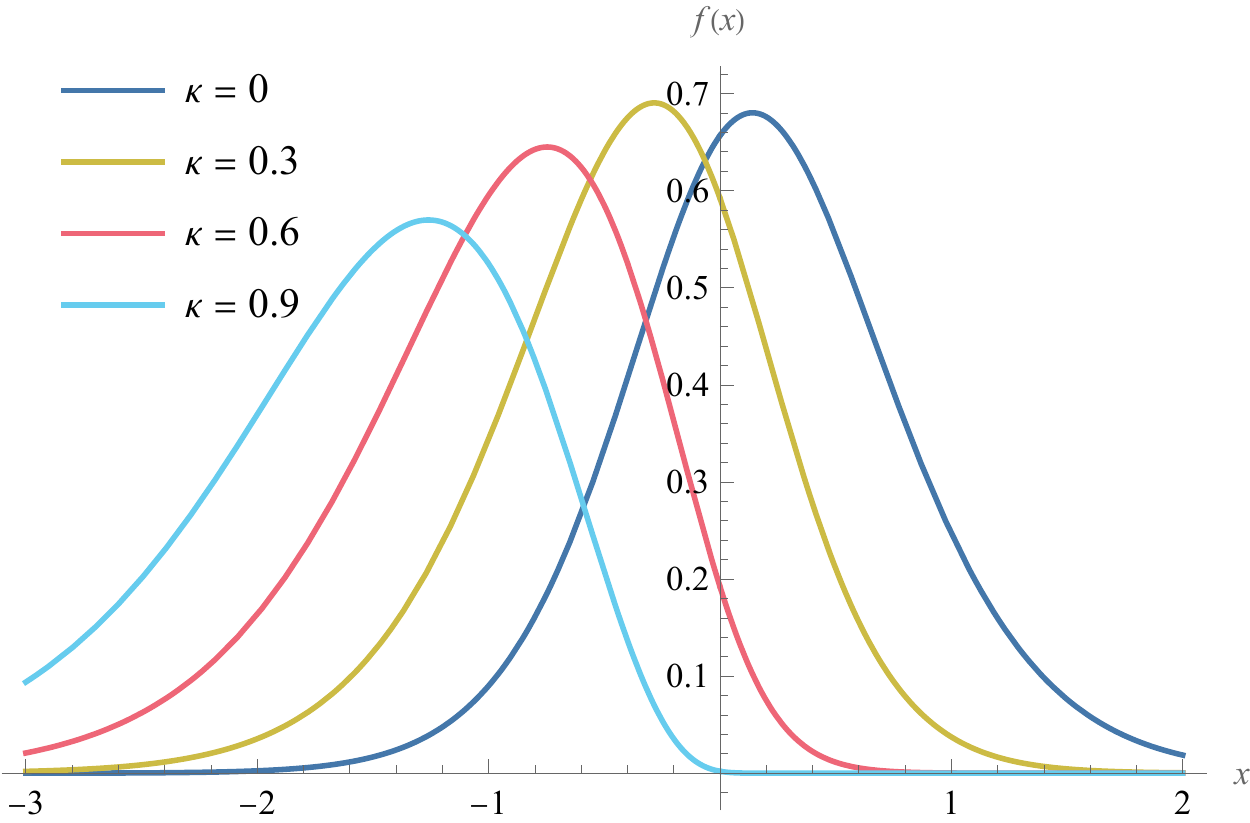}}
	\caption{PDFs of the NP detector when $N = 5$, $\rho = 0.1$, and $\kappa \in \{0, 0.3,\linebreak[1] 0.6, 0.9\}$.}
	\label{fig:PDF_kappa_5}
\end{figure}

Figs.\ \ref{fig:PDF_optimal_rho}--\ref{fig:PDF_kappa_5} show plots of the PDF of the NP detector for various values of $\rho$, $\kappa$, and $N$. In Figs.\ \ref{fig:PDF_optimal_rho} and \ref{fig:PDF_optimal_N}, the optimal case $\kappa = \rho$ is considered. In contrast, Figs.\ \ref{fig:PDF_kappa_1} and \ref{fig:PDF_kappa_5} show what happens when $\rho$ is fixed at 0.1 and $\kappa$ is varied. (The case where $\kappa$ is fixed and $\rho$ results in the same plots, but with a reflection in the $x$ axis.) Interestingly, these plots constitute a striking illustration of the central limit theorem: when $N = 1$, the PDF looks nothing like a normal distribution and even has a discontinuous derivative at $x = 0$, but when $N$ is as low as 5, the PDFs already approach the normal distribution.

\section{Receiver Operating Characteristic Curves}

The gold standard for evaluating radar detection performance is the \emph{receiver operating characteristic} (ROC) curve, which is effectively $\pd$ as a function of $\pfa$. Proposition \ref{prop:NP_optimal} implies that the ROC curve for $\bar{D}_\rho$ is greater than that of $\bar{D}_\kappa$ whenever $\kappa \neq \rho$. This behavior can be seen in Fig.\ \ref{fig:ROC_kappa}, which shows ROC curves for $\rho = 0.3$ and various values of $\kappa$. The plot shows us that the set of ROC curves is indeed bounded from above by the ROC curve for $\kappa = \rho$. However, the ROC curves do not change dramatically as $\kappa$ is varied from 0 to 0.6. This suggests that the optimal detector for the simple hypotheses \eqref{eq:hypotheses_simple} is still a viable detector when applied to the realistic hypotheses \eqref{eq:hypotheses}, even when $\kappa$ is chosen incorrectly.

Another way of seeing this is in Fig.\ \ref{fig:pd_kappa}, which plots $\pd/(\max \pd)$ as a function of $\kappa$; $\max \pd$ is the maximum $\pd$ over all $\kappa$ for fixed values of $\rho$, $N$, and $\pfa$. It is clear from these plots that the maximum $\pd$ occurs, as expected, when $\kappa = \rho$. It also appears that, as $N$ increases, $\pd$ deviates less and less from $\max \pd$.

\begin{figure}[t]
	\centerline{\includegraphics[width=\columnwidth]{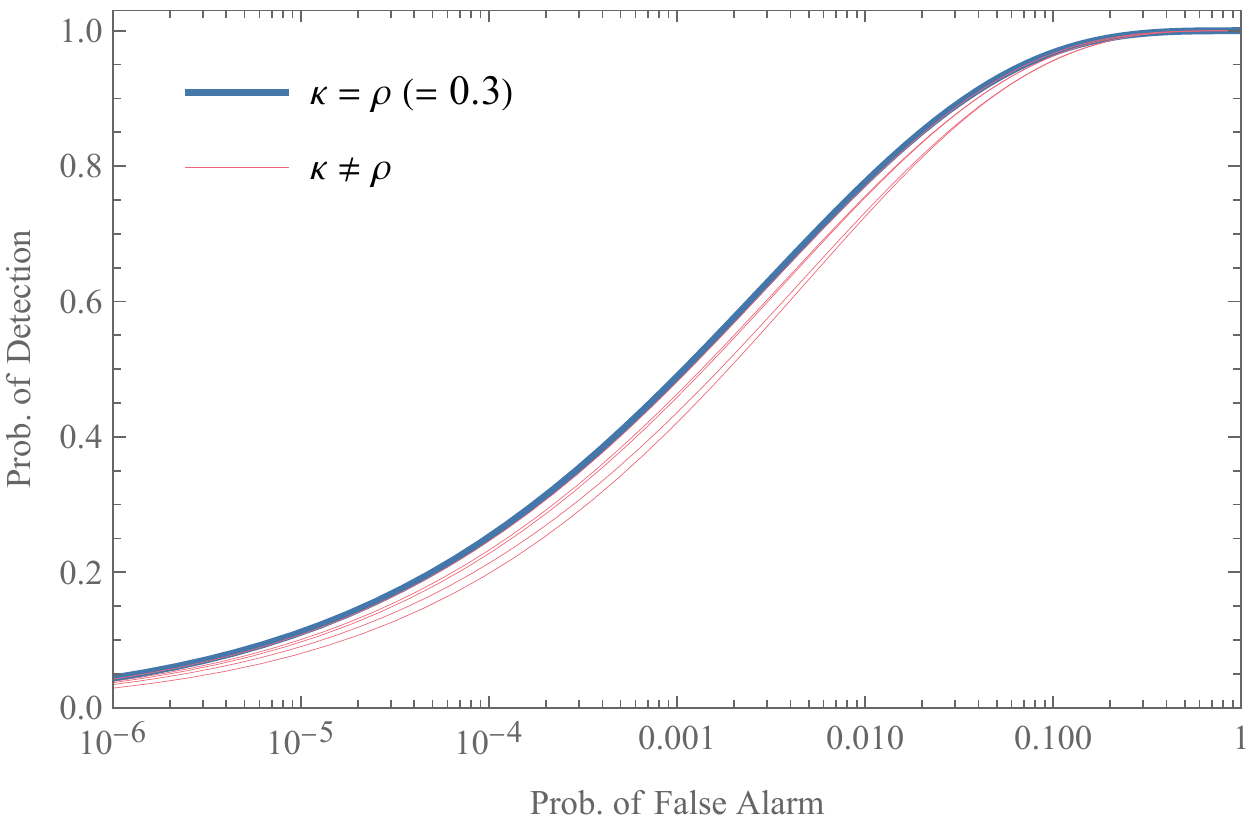}}
	\caption{ROC curves for the NP detector $\bar{D}_\kappa$ when $\rho = 0.3$, $N = 50$, and $\kappa \in \{0, 0.1, 0.2, \dots, 0.6\}$. The ROC curve for $\kappa = \rho = 0.3$ is drawn as a thick blue curve; the other ROC curves are thin red curves.}
	\label{fig:ROC_kappa}
\end{figure}

\begin{figure}[t]
	\centerline{\includegraphics[width=\columnwidth]{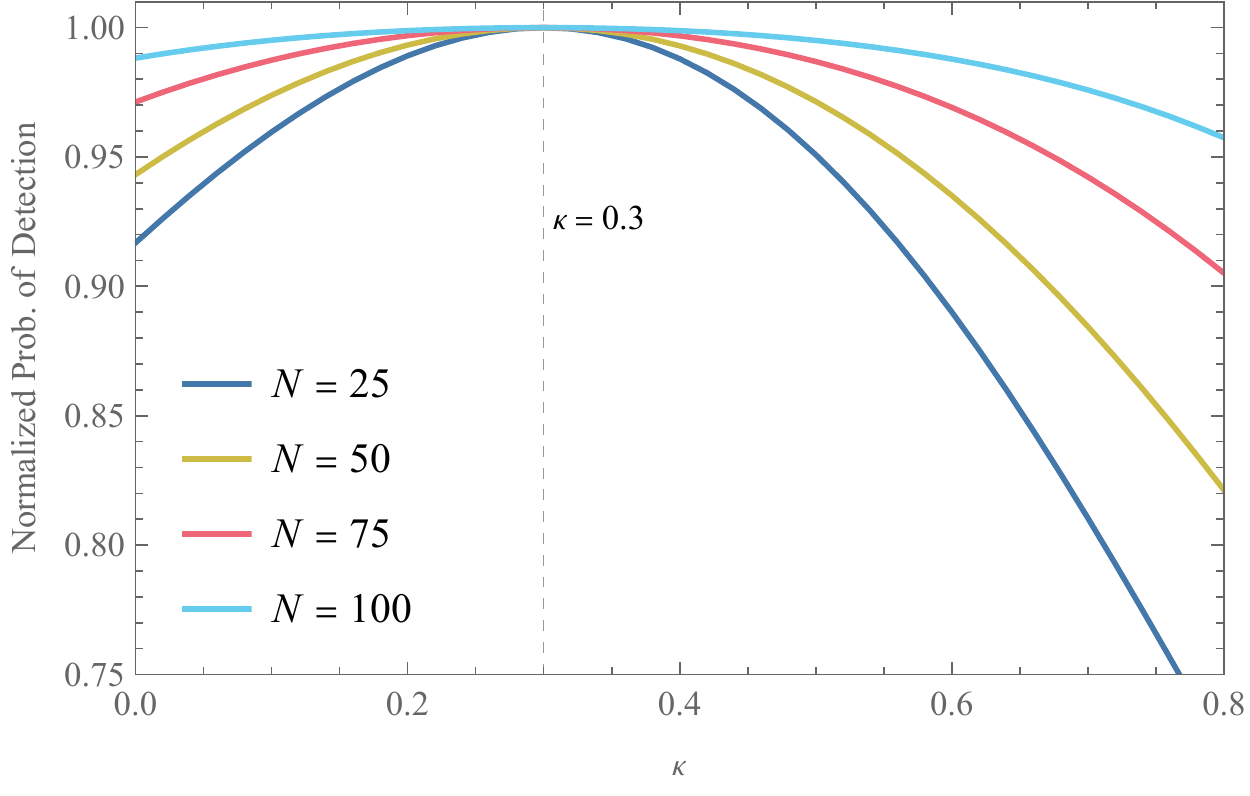}}
	\caption{Normalized probability of detection $\pd/(\max \pd)$ for the NP detector $\bar{D}_\kappa$ as a function of $\kappa$, where $\max \pd$ is the maximum $\pd$ when $\rho$, $N$, and $\pfa$ are fixed. We used $\rho = 0.3$, $\pfa = 10^{-2}$, and $N \in \{25, 50, 75, 100\}$. Note the peaks at $\kappa = \rho$.}
	\label{fig:pd_kappa}
\end{figure}

Because a closed-form expression for the cumulative density function of the variance-gamma distribution is not known, the ROC curves in Fig.\ \ref{fig:ROC_kappa} were generated numerically. When $N$ is large, however, we may obtain an approximate ROC curve formula by invoking the central limit theorem. This leads to the following proposition.

\begin{proposition} \label{prop:ROC_approx}
	In the limit of large $N$, the ROC curve for the NP detector is approximately
	\begin{equation} \label{eq:ROC_approx}
		p_\mathit{d}(p_\mathit{fa}) = \frac{1}{2} \erfc \mleft( \frac{ \sqrt{1 + \kappa^2}\erfc^{-1}(2 p_\mathit{fa}) - \sqrt{N} \rho }{ \sqrt{(\rho - \kappa)^2 + (1 - \rho \kappa)^2} } \mright)
	\end{equation}
	where $\erfc(\cdot)$ is the complementary error function.
\end{proposition}

\begin{proof}
	Note that the NP detector is the sample mean of $N$ independent and identically distributed random variables. Therefore, we can invoke the central limit theorem to state that
	\begin{equation}
		\bar{D}_\kappa \sim \mathcal{N} \mleft( 2(\rho - \kappa) , \frac{2(\rho - \kappa)^2 + 2(1 - \rho \kappa)^2}{N} \mright)
	\end{equation}
	in the limit of large $N$. The mean and variance were calculated using \eqref{eq:VG_mean} and \eqref{eq:VG_var}; they could also have been calculated directly from \eqref{eq:det_NP_char_func}.
	
	When no radar target is present, $\rho = 0$. It follows that the probability of false alarm in the large-$N$ limit, given a detection threshold $T$, is
	\begin{equation}
		p_\mathit{fa} = \frac{1}{2} \erfc \mleft( \dfrac{ (T + 2 \kappa) \sqrt{N} }{ 2 \sqrt{1 + \kappa^2} } \mright).
	\end{equation}
	Solving for $T$ yields
	\begin{equation}
		T = 2 \sqrt{\frac{1 + \kappa^2}{N}} \erfc^{-1}(2 p_\mathit{fa}) - 2 \kappa.
	\end{equation}
	Similarly, the probability of detection is
	\begin{equation}
		p_\mathit{d} = \frac{1}{2} \erfc \mleft( \frac{ \sqrt{N}[T - 2(\rho - \kappa)] }{ 2 \sqrt{(\rho - \kappa)^2 + 2(1 - \rho \kappa)^2} } \mright).
	\end{equation}
	Substituting $T$ into this expression gives us \eqref{eq:ROC_approx}.
\end{proof}

\begin{figure}[t]
	\centerline{\includegraphics[width=\columnwidth]{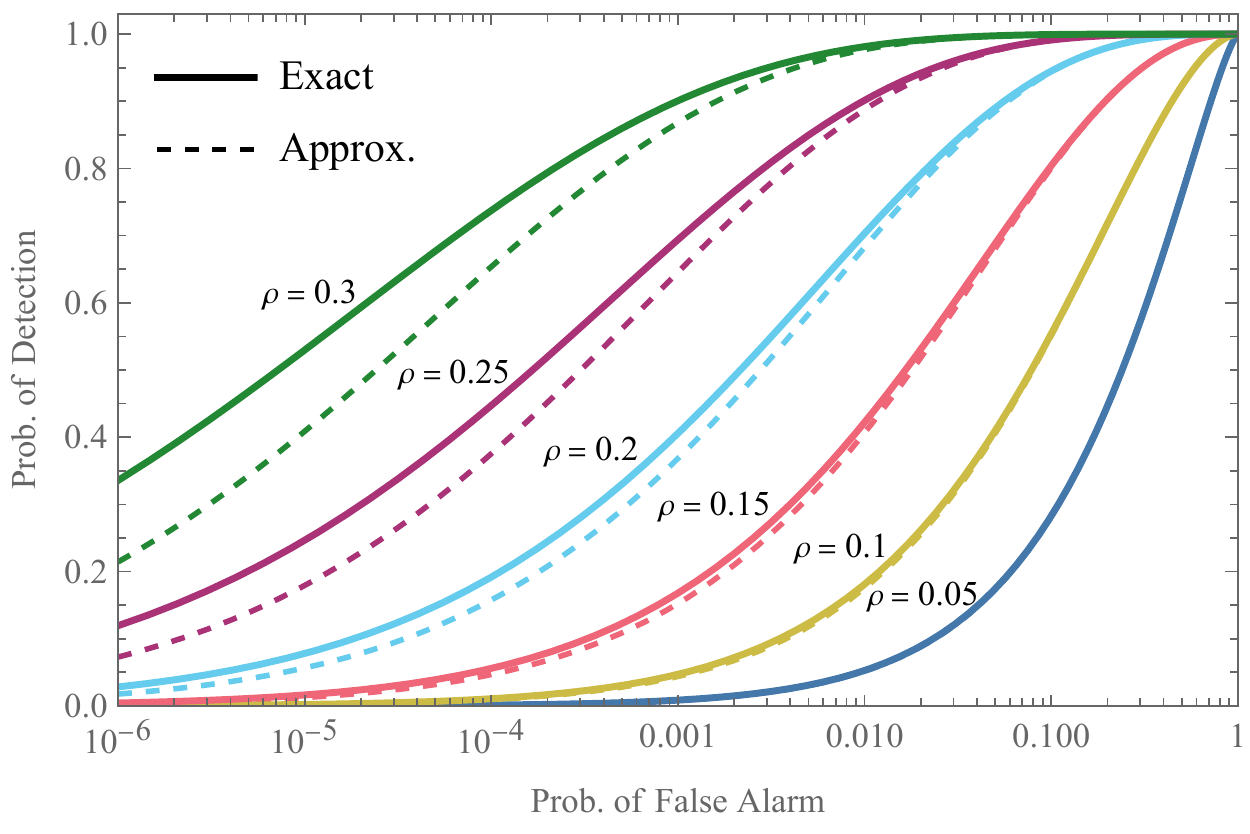}}
	\caption{ROC curves for the NP detector $\bar{D}_\kappa$ when $N = 100$, $\kappa = \rho$, and $\rho \in \linebreak[1] \{0.05, 0.1, \dots, 0.3\}$. Dashed curves indicate approximate ROC curves calculated using \eqref{eq:ROC_approx}.}
	\label{fig:ROC_ideal_rho}
\end{figure}

\begin{figure}[t]
	\centerline{\includegraphics[width=\columnwidth]{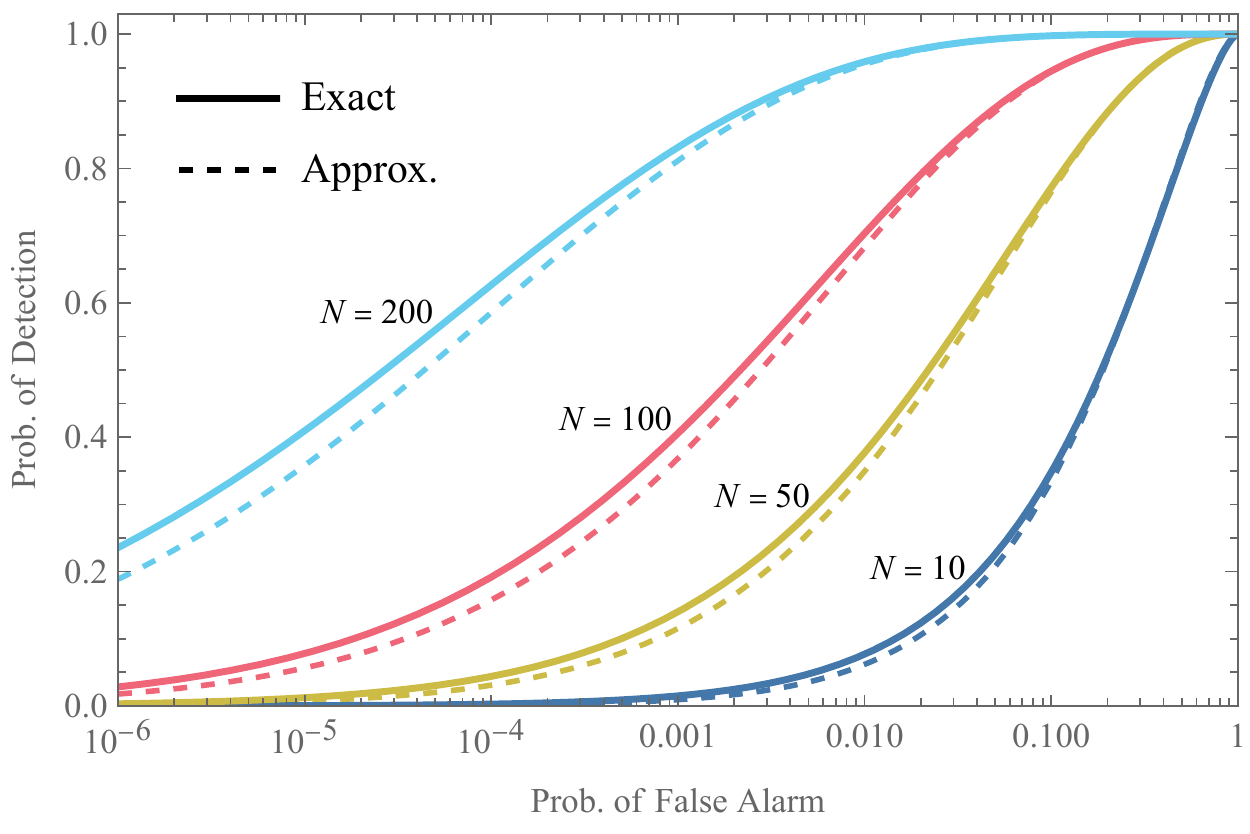}}
	\caption{ROC curves for the NP detector $\bar{D}_\kappa$ when $\kappa = \rho = 0.2$ and $N \in \{10,\linebreak[1] 50,\linebreak[1] 100, 200\}$. Dashed curves indicate approximate ROC curves calculated using \eqref{eq:ROC_approx}.}
	\label{fig:ROC_ideal_N}
\end{figure}

Figs.\ \ref{fig:ROC_ideal_rho} and \ref{fig:ROC_ideal_N} show ROC curves for the NP detector in the ideal case where $\kappa = \rho$. As discussed above, these ROC curves are optimal for any radar that operates as described in Sec. \ref{sec:nr_target_det}. Plotted together with these ROC curves are approximations calculated using \eqref{eq:ROC_approx}. It can be seen that the approximation gives a good general idea of the ROC curve behavior of $\bar{D}_\rho$. However, it is not always perfect. The approximate formula does err noticeably sometimes, especially when $\rho$ is high. However, it is much easier to calculate than the exact ROC curve, which requires numerical integration of \eqref{eq:VG_pdf}. This integration is sometimes numerically unstable, especially for large $N$.

It is interesting to note that the approximate expression \eqref{eq:ROC_approx} does not correctly reflect the optimality of the NP detector. When \eqref{eq:ROC_approx} is maximized over $\kappa$, the maximum does not occur at $\kappa = \rho$. This does not contradict Prop.\ \ref{prop:NP_optimal}, of course, because \eqref{eq:ROC_approx} is only an approximation to the ROC curve. We point this out merely as an unfortunate consequence of the assumptions used to derive \eqref{eq:ROC_approx}. It may be that this behavior is related to that seen in Fig.\ \ref{fig:pd_kappa}, where the relative decrease in $\pd$ when an incorrect $\kappa$ is chosen becomes smaller and smaller as $N$ is increased.

\section{The Low-Correlation Limit}

In \cite{luong2022performance}, we showed that $\rho$ depends on the distance $R$ from the radar to the target through the expression 
\begin{equation} \label{eq:rho_range}
	\rho(R) = \frac{\rho_0}{\sqrt{1 + (R/R_c)^4}},
\end{equation}
where $\rho_0$ is a constant and $R_c$ is a characteristic length that depends on the parameters in the radar range equation. This equation shows that, at long ranges, the correlation coefficient goes to zero. This gives us a physical motivation for considering the limit $\rho \to 0$: this limit corresponds to distant targets.

In the small-$\rho$ limit, it is reasonable to consider the detector $\bar{D}_0$, the NP detector with $\kappa = 0$. In fact, this detector was previously studied in \cite{luong2019roc,luong2020simdet} (under the name $D_1$). Moreover, the ``digital receiver'' employed in the quantum radar experiment by Barzanjeh et al.\ \cite{barzanjeh2019experimental} is equivalent to $\bar{D}_0$ in the classical limit, as we showed in \cite{luong2021when}. In the light of the results in this paper, we can also motivate $\bar{D}_0$ as being the most natural choice of detector, among all the members of the NP detector family, when $\rho$ is small. More than that, it is the locally most powerful test for small deviations from the null hypothesis $\rho = 0$---which is exactly the case that we are considering in this section. And $\bar{D}_0$ can be shown to be equivalent to the score test (Lagrange multiplier test), which is known to be locally most powerful \cite{rao2005score}.

\begin{figure}[t]
	\centerline{\includegraphics[width=\columnwidth]{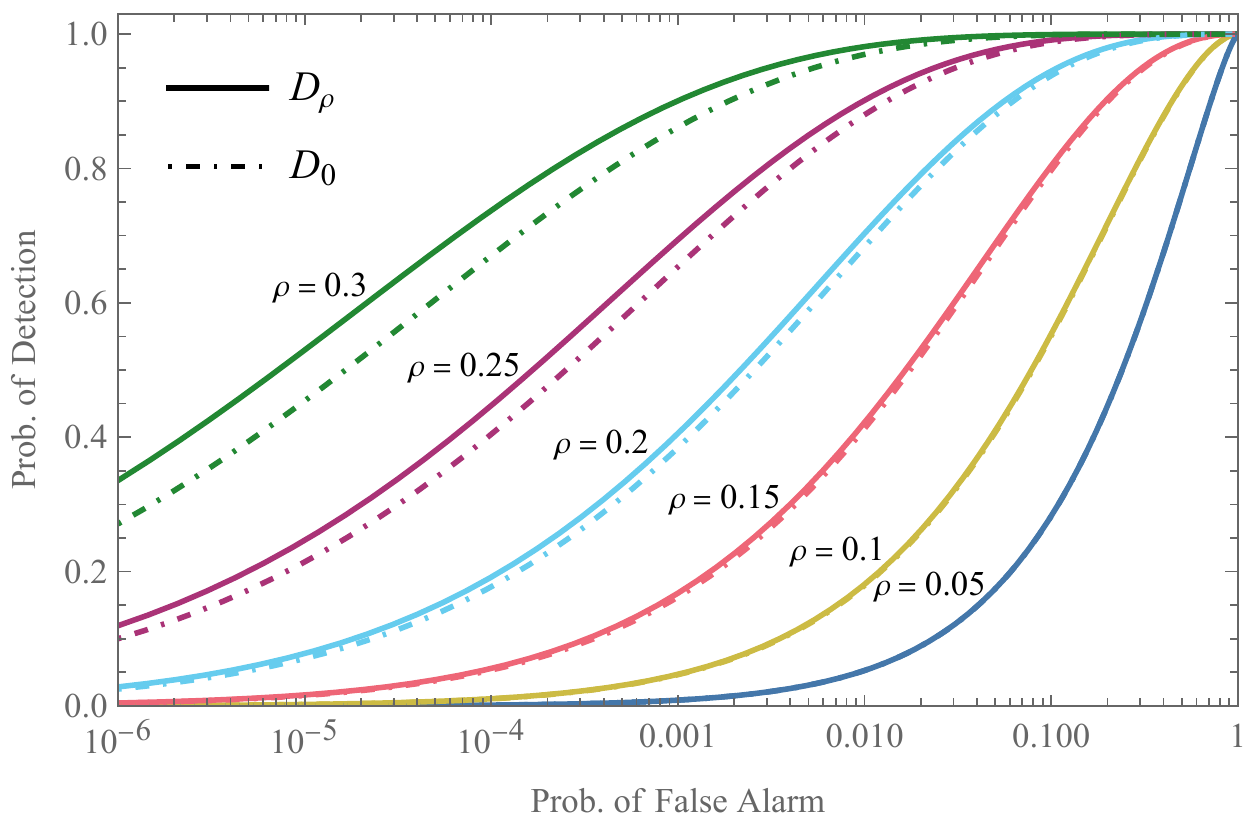}}
	\caption{Comparison of ROC curves for $\bar{D}_\rho$ (solid lines) and $\bar{D}_0$ (dot-dashed lines) when $N = 100$ and $\rho \in \linebreak[1] \{0.05, 0.1, \dots, 0.3\}$.}
	\label{fig:ROC_ideal_vs_0_rho}
\end{figure}

\begin{figure}[t]
	\centerline{\includegraphics[width=\columnwidth]{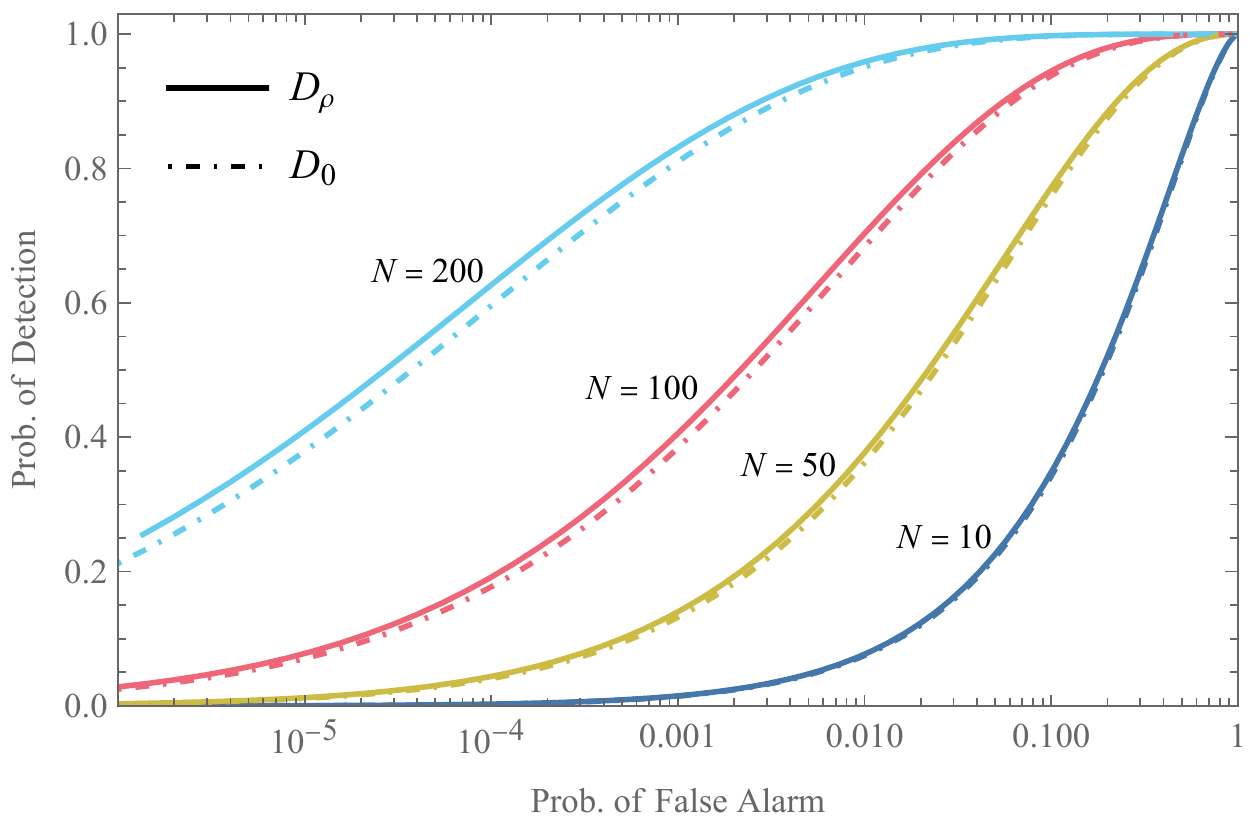}}
	\caption{Comparison of ROC curves for $\bar{D}_\rho$ (solid lines) and $\bar{D}_0$ (dot-dashed lines) when $\kappa = \rho = 0.2$ and $N \in \{10,\linebreak[1] 50,\linebreak[1] 100\}$.}
	\label{fig:ROC_ideal_vs_0_N}
\end{figure}

In Figs.\ \ref{fig:ROC_ideal_vs_0_rho} and \ref{fig:ROC_ideal_vs_0_N}, we compare the ROC curve performance of $\bar{D}_0$ with that of the optimal detector $\bar{D}_\rho$. It is evident that the difference in performance is hardly noticeable when $\rho$ is small, as expected. When $\rho$ is as large as 0.3, the difference is more noticeable, but still not very great. We may conclude that $\bar{D}_0$ is a good choice of detector, particularly when $\rho$ is small (as is usually the case). And unlike the optimal detector $\bar{D}_\rho$, it is a practical detector because it requires no foreknowledge of $\rho$. 

Recall that, in order to derive the NP detector, we made assumptions on received and reference signal powers, as well as the relative phase between them, leading to the simplified covariance matrix \eqref{eq:cov_simplified}. Since $\bar{D}_0$ could potentially be used in practice, it is of interest to derive the distribution of $\bar{D}_0$ without making these assumptions. The following proposition shows that $\bar{D}_0$ is variance-gamma distributed even when the general covariance matrix \eqref{eq:QTMS_cov} is used.

\begin{proposition}
	When the radar signals are described by the covariance matrix \eqref{eq:QTMS_cov}, the detector $\bar{D}_0$ has the distribution
	\begin{equation} \label{eq:det0_dist}
		\bar{D}_0 \sim \mathit{VG} \mleft( 0, \sigma_1\sigma_2 \sqrt{\frac{2(1-\rho^2 \cos^2 \phi)}{N}}, 2\rho\sigma_1\sigma_2 \cos \phi, \frac{1}{N} \mright).
	\end{equation}
\end{proposition}

\begin{proof}
	In \cite{gaunt2018note}, Gaunt showed that the product of two zero-mean normal random variables is variance-gamma distributed, as is the sample mean of such a product; see also \cite{nadarajah2016distribution}. (Note that Gaunt used a different parameterization of the variance-gamma distribution.) This exactly fits the case of $I_1 I_2$ and $Q_1 Q_2$. Using the fact that the Pearson correlation coefficient between $I_1$ and $I_2$ is $\rho \cos \phi$, \cite[Theorem 1]{gaunt2018note} yields
	\begin{equation}
		\overline{I_1 I_2} \sim \mathit{VG} \mleft( 0, \sigma_1\sigma_2 \sqrt{\frac{1-\rho^2 \cos^2 \phi}{N}}, \rho\sigma_1\sigma_2 \cos \phi, \frac{2}{N} \mright).
	\end{equation}
	By inspection of \eqref{eq:QTMS_cov}, we see that $\pm\overline{Q_1 Q_2}$ is distributed identically to $\overline{I_1 I_2}$, assuming that the same sign as \eqref{eq:det0} is chosen. Then \eqref{eq:det0_dist} follows upon applying \cite[Eq.\ 4]{gaunt2018note}, which gives the distribution of the sum of variance-gamma variates.
\end{proof}


In the same vein as Prop.\ \ref{prop:ROC_approx}, we can exploit the central limit theorem to give an approximate expression for the ROC curve of $\bar{D}_0$.

\begin{proposition} \label{prop:ROC_0_approx}
	In the limit of large $N$, the ROC curve for $\bar{D}_0$ is approximately
	\begin{equation} \label{eq:ROC_0_approx}
		p_\mathit{d}(p_\mathit{fa}) = \frac{1}{2} \erfc \mleft( \frac{ \erfc^{-1}(2 p_\mathit{fa}) - \sqrt{N} \rho \cos \phi }{ \sqrt{1 + \rho^2 \cos^2 \phi} } \mright).
	\end{equation}
\end{proposition}

\begin{proof}
	The proof is essentially the same as that of Prop.\ \ref{prop:ROC_approx}.
\end{proof}

Although the detection performance does not depend on $\sigma_1$ or $\sigma_2$, it does depend on $\phi$. Thus, $\bar{D}_0$ is a phase-dependent detector: its detection performance when $\phi \neq 0$ is equivalent to reducing $\rho$ to $\rho \cos \phi$. There are situations (e.g.\ in the biomedical realm) where the phase can potentially be known beforehand, so the practicality of $\bar{D}_0$ is not fully ruled out. For a more thorough discussion of these issues, we refer the reader to  \cite{luong2022likelihood}.

\section{Conclusion}

In this paper, we derived a family of detectors for QTMS and noise radars based on the Neyman-Pearson lemma. This family of detectors is parameterized by a single parameter $\kappa$ such that, when $\kappa$ equals the correlation coefficient $\rho$, the resulting detector is optimal. Although the simple hypotheses \eqref{eq:hypotheses_simple} used to derive the detectors are not applicable in practice, they allow us to set an upper bound on the ROC curve of any radar satisfying the assumptions made in Sec.\ \ref{sec:nr_target_det}. Our work thus furnishes a natural benchmark against which any other detector can be compared.

Since the correlation coefficient $\rho$ is often small, it is reasonable to ask what happens in the $\rho \to 0$ limit. We found that $\bar{D}_0$, the NP detector when $\kappa = 0$, is a good detector in this limit. This is not surprising, since it also happens to be the locally most powerful test for values of $\rho$ close to 0. Unlike the optimal NP detector, which requires $\rho$ to be known beforehand, $\bar{D}_0$ requires no such foreknowledge and is therefore a practical detector. As a matter of fact, $\bar{D}_0$ has already been used in actual experiments; the work in this paper provides an additional motivation for using this detector (at least when phase is not an issue).

One of the motivations for studying the simple hypotheses \eqref{eq:hypotheses_simple}, apart from deriving an upper bound on detection performance, is that some of the early QI literature implicitly assumes that $\rho$ is fixed at the known value $\kappa$. Strictly speaking, this is incorrect; the hypotheses \eqref{eq:hypotheses} should have been used instead. In this paper, however, we saw that an incorrect choice of $\kappa$ does not necessarily affect the detection performance very gravely. However, our work does not fully apply to standard quantum illumination radars, which use a complicated quantum-optical receiver instead of heterodyne receivers. In future work, it would be interesting to see whether our conclusions can be extended to standard quantum illumination.

\section*{Acknowledgment}

This work was supported by the Natural Science and Engineering Research Council of Canada (NSERC). D.\ Luong also acknowledges the support of a Vanier Canada Graduate Scholarship.

\bibliographystyle{ieeetran}
\bibliography{qradar_refs,own_refs}
	
\end{document}